\documentclass[conference, a4paper]{IEEEtran}
\IEEEoverridecommandlockouts

\usepackage{cite}
\usepackage{amsmath,amssymb,amsfonts}
\usepackage{algorithmic}
\usepackage{amsthm}
\usepackage{epsfig}
\usepackage{color}
\usepackage{derivative}
\usepackage{lettrine}
\usepackage{float}
\usepackage{lipsum}
\usepackage{bm}
\usepackage{mathtools}
\usepackage{bbm}
\usepackage{suffix}
\usepackage[colorlinks]{hyperref} 
\usepackage[nameinlink,noabbrev]{cleveref}
\usepackage{textcomp}
\usepackage{xcolor}
\usepackage{stackengine}
\usepackage[ruled,norelsize]{algorithm2e}
\def \treq {\stackrel{\tiny \Delta}{=}}
\usepackage{multicol}
\usepackage{euler}

\newcommand{\Q}{\ensuremath{\mathrm{Q}}}
\newcommand{\Qi}{\ensuremath{\mathrm Q}^{-1}}
\newcommand{\e}{\mathbf{e}}

\newcommand{\E}{\ensuremath{\mathbb E}}
\renewcommand{\Pr}{\ensuremath{\mathbb P}}
\newcommand{\R}{\ensuremath{\mathbb R}}

\newcommand{\CN}{\ensuremath{\mathcal {CN}}}

\newcommand{\bgamma}{\ensuremath{\bar{\gamma}}}
\newcommand{\bdelta}{\ensuremath{\bar{\delta}}}

\newcommand{\Exp}[1]{\mathcal{E}\left(#1\right)}

\newcommand{\Prob}[1]{(\textbf{P}{#1})}

\def \treq {\stackrel{\tiny \Delta}{=}}

\newtheorem{prop}{Proposition}
\newtheorem{cor}{Corollary}
\newtheorem{remark}{Remark}

\usepackage{hyperref} 

\hypersetup{
    colorlinks=true,
    linkcolor=red,
    urlcolor=red,
    linktoc=all,
    citecolor=blue
           }

\makeatletter
\def\ps@IEEEtitlepagestyle{%
  \def\@oddfoot{\mycopyrightnotice}%
  \def\@oddhead{\hbox{}\@IEEEheaderstyle\leftmark\hfil\thepage}\relax
  \def\@evenhead{\@IEEEheaderstyle\thepage\hfil\leftmark\hbox{}}\relax
  \def\@evenfoot{}%
}
\def\mycopyrightnotice{%
  \begin{minipage}{\textwidth}
  \centering \scriptsize
  \copyright~2024 IEEE. Personal use of this material is permitted.  Permission from IEEE must be obtained for all other uses, in any current or future media, including reprinting/republishing this material for advertising or promotional purposes, creating new collective works, for resale or redistribution to servers or lists, or reuse of any copyrighted component of this work in other works.
  \end{minipage}
}
\makeatother
\begin{document}
\bstctlcite{IEEEexample:BSTcontrol}

\title{On the Information Leakage Performance of \\Secure Finite Blocklength Transmissions over Rayleigh Fading Channels}



 \author{\IEEEauthorblockN{Milad Tatar Mamaghani\IEEEauthorrefmark{1}, Xiangyun Zhou\IEEEauthorrefmark{1}, Nan Yang\IEEEauthorrefmark{1}, A. Lee Swindlehurst\IEEEauthorrefmark{2},  and H. Vincent Poor\IEEEauthorrefmark{3}}
\IEEEauthorblockA{
\IEEEauthorrefmark{1}School of Engineering, Australian National University, Canberra, ACT 2601, Australia\\
\IEEEauthorrefmark{2}Henry Samueli School of Engineering, University of California, Irvine, CA 92697, USA\\
\IEEEauthorrefmark{3}Department of Electrical Engineering, Princeton University, Princeton, NJ 08544, USA}
Email:~\{\href{mailto:milad.tatarmamaghani@anu.edu.au}{\textcolor{black}{milad.tatarmamaghani}}, \href{mailto:xiangyun.zhou@anu.edu.au}{\textcolor{black}{xiangyun.zhou}}, \href{mailto:nan.yang@anu.edu.au}{\textcolor{black}{nan.yang}}\}@anu.edu.au,   \href{mailto:swindle@uci.edu}{\textcolor{black}{swindle@uci.edu}},
\href{mailto:poor@princeton.edu}{\textcolor{black}{poor@princeton.edu}}
 \thanks{An extended version of this work has been submitted \cite{mamaghani2023performance}.}}

\maketitle

\begin{abstract}
This paper presents a secrecy performance study of a wiretap communication system with finite blocklength (FBL) transmissions over Rayleigh fading channels, based on the definition of an average information leakage (AIL) metric. We evaluate the exact and closed-form approximate AIL performance, assuming that only statistical channel state information (CSI) of the eavesdropping link is available. Then, we reveal an inherent statistical relationship between the AIL metric in the FBL regime and the commonly-used secrecy outage probability in conventional infinite blocklength communications. Aiming to improve the secure communication performance of the considered system, we formulate a blocklength optimization problem and solve it via a low-complexity approach. Next, we present numerical results to verify our analytical findings and provide various important insights into the impacts of system parameters on the AIL. Specifically, our results indicate that i) compromising a small amount of AIL can lead to significant reliability improvements, and ii) the AIL experiences a secrecy floor in the high signal-to-noise ratio regime.
\end{abstract}


\IEEEpeerreviewmaketitle

\section{Introduction}
In today's interconnected world, wireless networks have permeated numerous applications, becoming an indispensable aspect of our lives. From safeguarding private information to powering essential services such as credit data, e-health records, or vital command and control messages, the significance of beyond-5G (B5G) wireless networks underscores the paramount need for robust security measures \cite{Wu2018survey}. Thus, communication security has received considerable attention from both academia and industry, particularly at the physical layer.   Physical-layer security (PLS) is a promising security candidate that exploits the distinct features and randomness of the transmission medium such as channel impairments, noise, or smart signaling to protect wireless transmissions. PLS can enhance confidentiality, decrease reliance on upper-layer cryptography, and ensure security without the need for sophisticated key-exchange procedures \cite{Poor2017}. 

Conventional PLS developments have been centered around the idea of \textit{secrecy capacity}, which refers to the maximum achievable secure rate that guarantees both reliability and confidentiality over a wiretap channel. Wyner in  \cite{Wyner1975} showed that by employing so-called wiretap coding, it is possible to concurrently minimize the decoding error probability at a legitimate receiver and reduce the information leakage to a malicious adversary to an arbitrarily low level with an infinitely long coding blocklength. Nevertheless, emerging scenarios in B5G  wireless systems such as machine-type communication (MTC) require data traffic characterized by short packets in order to satisfy the broader communication requirements \cite{Akyildiz2020,  Bockelmann2016}. The use of short packets calls for finite block-length (FBL) analysis, as the traditional asymptotic analysis no longer holds in this regime. While leveraging FBL communication helps minimize end-to-end transmission latency due to a reduction in the number of channel uses, it generally comes with a decrease in channel coding gain, making it challenging to ensure communication reliability as well as secrecy. In addition, since wireless FBL communication scenarios cannot be accommodated by traditional PLS designs, which rely on the infinite blocklength (IBL) assumption, it is crucial to meticulously develop PLS schemes tailored to the specific requirements of the FBL regime.

Recently, some research has explored the limitations of FBL transmissions from different perspectives via information-theoretic approaches. Polyanskiy \textit{et al.} in \cite{Polyanskiy2010} addressed the problem of maximizing the channel coding rate in the FBL regime with given reliability constraints in general communication channels.  This work urged the research community to further explore the characterization of non-asymptotic achievable rate regions in different non-security-based schemes \cite{Mary2016c} and security-based scenarios \cite{Wang2019e, Chen2020b, mamaghani2023secure, Feng2022}, to determine the practical impacts of FBL on wireless communications. In particular,  the authors in \cite{Wang2019e} studied secure FBL communication for mission-critical Internet of Things (IoT) applications with an external eavesdropper.  The work in \cite{Chen2020b} investigated secrecy performance in cognitive IoT with low-latency and security requirements. In  \cite{mamaghani2023secure}, the design of a secure aerial communication system with FBL was considered to improve the average secrecy rate while meeting security and reliability requirements. In \cite{Feng2022}, the authors presented an analytical framework to investigate the average secrecy throughput in the FBL regime.

Despite the aforementioned research efforts, there are still some fundamental issues that need to be addressed when designing a secure communication system operating with FBL transmissions. For instance, the information leakage performance of such systems with fading channels has been less reported in the existing literature, e.g., \cite{Wang2019e, Chen2020b, mamaghani2023secure, Feng2022, Zheng2020, oh2023joint}. Nevertheless, legitimate users do not have direct control over the amount of information leakage because the instantaneous channel state information (CSI) for the eavesdropping links is often unknown in fading channels. Information leakage is essentially a random quantity that fluctuates depending on the eavesdropper's channel quality. Therefore, it is necessary to conduct further research in order to establish a statistical measure of information leakage in the FBL regime. This work presents an analytical framework to address the problem of information leakage performance through FBL transmissions over fading channels, which is pivotal for guiding practical designs. We validate the theoretical findings via simulations, and also design the coding blocklength for system secrecy performance improvement via different approaches. In particular, we obtain a closed-form analytical expression for the optimal blocklength over Rayleigh fading channels under a constraint on the maximum allowable AIL.

\section{System Model and Assumptions}\label{sec:sysmodel}


Consider a typical downlink wiretap communication system, where an access point (Alice) communicates with an intended receiver (Bob) via confidential FBL transmissions while an adversary (Eve) attempts to wiretap the ongoing legitimate communication. We assume that each node in the network is equipped with a single antenna. We also assume that communication channels undergo quasi-static  Rayleigh block fading, where the channel coefficients remain unchanged over the duration of one FBL packet transmission within the channel coherence time, and their amplitudes are independent and identically distributed (i.i.d.) according to a Rayleigh distribution from one packet to another. We denote $h_b$  and $h_e$ as the reciprocal complex-valued Alice-Bob and Alice-Eve channels, respectively, considering both large-scale attenuation and small-scale fading. In this work, we assume that the CSI for the main link, $h_b$, is perfectly known to both Alice and Bob via channel reciprocity and training. We also assume that Eve can obtain the instantaneous CSI for her channel, $h_e$. However, the instantaneous CSI of $h_e$ is unavailable to Alice and Bob, and they are only aware of the statistics of $h_e$, due to Eve's use of passive wiretapping.

Alice generates the transmit signal as a unit-power waveform $s$, i.e.,  $\E\{\|s\|^2\}=1$, where $\E\{\cdot\}$ indicates the expectation operator, and broadcasts it with transmit power $P$ over $N$ channel uses.  Assuming that communication channels are corrupted with additive white Gaussian noise (AWGN),  the received signal at node $i$, denoted by $y_i$, can be described as
\begin{align}\label{yj}
    y_i = \sqrt{P}  h_i s  + \xi_i, \quad i\in \{b, e\} 
\end{align}
where  $h_i$ indicates the $i$-th channel coefficient such that $\E\{\|h_i\|^2\}=\mu_i$, and $\xi_i \sim \CN(0, \sigma^2)$ accounts for the AWGN at the $i$-th receiver, where $\CN(\mu, \sigma^2)$ denotes a circularly symmetric complex normal distribution with mean $\mu$ and variance $\sigma^2$. Accordingly,  the received signal-to-noise
ratio (SNR) at the intended receiver $i$ can be obtained as ${\gamma}_i = \rho \|h_i\|^2$, and the transmit SNR is defined as $\rho=\frac{p}{\sigma^2}$, which follows an exponential distribution with mean $\bgamma_i\treq\rho \mu_i$, i.e., $\gamma_i\sim \Exp{\bgamma_i}$ with  
\begin{align}\label{gamma}
        f_{\gamma_i}(x)= \frac{1}{\bgamma_i}\e^{{-\frac{x}{\bgamma_i}}},\quad 
        F_{\gamma_i}(x) = 1-\e^{{-\frac{x}{\bgamma_i}}},~x>0
\end{align}
where $f_{\gamma_i}(\cdot)$ and $F_{\gamma_i}(\cdot)$ represent the probability density function (PDF) and the cumulative distribution function (CDF) of the random variable (r.v.) $\gamma_i$, respectively.

\section{Secrecy Performance}\label{sec:performance}
In this section, we focus on the ergodic secrecy performance of FBL transmissions quantified by the information leakage to Eve, while assuming the legitimate link can accommodate a given desired level of reliability. Furthermore, we aim to establish a simple and intuitive statistical relationship between the secrecy outage formulation widely used in the IBL regime and the variational distance-based information leakage formulation for FBL transmissions. Our objective is to exploit this inherent relationship to guide secure communication designs in the FBL regime when dealing with fading channels. 

\subsection{Average Information Leakage}
Unlike the conventional IBL regime, both decoding errors at Bob and information leakage to Eve may occur in the FBL regime, which leads to a further loss of communication reliability and secrecy. Accordingly, for the considered FBL communication system, the achievable secrecy rate in bits-per-channel-use (bpcu), while sustaining a desired decoding error probability $\varepsilon$ at Bob and information leakage $\delta$ to Eve, is approximately given by \cite{Yang2019}
\begin{align}\label{sp_secrate}
    R^*_s \approx \Big[ C_s(\gamma_b, \gamma_e) - \sqrt{\frac{V_b}{N}}\Qi\left(\varepsilon\right) -\sqrt{\frac{V_e}{N}}\Qi\left(\delta\right)\Big]^+,
\end{align}
where $[x]^+=\max\{x,0\}$, and $\Q^{-1}(x)$ is the inverse of the Gaussian Q-function defined as $\Q(x)=\int^{\infty}_{x}\frac{1}{\sqrt{2\pi}}\e^{-\frac{r^2}{2}}dr$. Additionally, $C_s$ is the secrecy capacity in the IBL regime given by
\begin{align}
    C_s(\gamma_b, \gamma_e) = \log_2(1+ \gamma_b) - \log_2(1+\gamma_e),
\end{align}
and $V_i$ in \eqref{sp_secrate} represents the stochastic variation of the $i$-th channel, which can be expressed as
\begin{align}
    V_i = \log^2_2\e~\frac{\gamma_i(\gamma_i+2)}{(\gamma_i+1)^2},\quad i\in \{b, e\}.
\end{align}

Since Alice knows the instantaneous CSI of Bob's but not Eve's channel, the legitimate users have the ability to control the reliability performance ($\varepsilon$) but not the secrecy performance ($\delta$). Accordingly, we assume that the desired reliability for the given $\varepsilon$ can be achieved, and thus focus our efforts on determining the best attainable secrecy performance characterized by $\delta$ for the considered scenario.  To this end, we assume that for each FBL transmission, Alice sends $m$ information bits over $N$ channel uses, such that the secrecy rate is given by $R_s=\frac{m}{N}$ in the considered wiretap system. To obtain the best achievable secrecy performance in terms of information leakage, we set $R_s = R^*_s$, where $R^*_s$ is the achievable FBL secrecy rate given by \eqref{sp_secrate},  and conduct some mathematical manipulations leading to
\begin{align}\label{delta}
\delta =\Q\left(\sqrt{\frac{N}{V_e}}\left[C_s(\gamma_b, \gamma_e)- \sqrt{\frac{V_b}{N}}\Qi(\varepsilon) -\frac{m}{N}\right]\right).
\end{align}
We note that \eqref{delta} essentially indicates the amount of information leakage that would occur given the best possible coding strategy for transmitting the confidential information at the rate of $R_s$ with a decoding error probability of $\varepsilon$ at Bob. In addition, it is evident from \eqref{delta} that $\delta$ is a random quantity determined by the instantaneous SNR of Eve's channel $\gamma_e$. Since only statistical information about Eve's channel is available, it seems reasonable to investigate the ergodic secrecy performance by averaging the information leakage $\delta$ over all realizations of $\gamma_e$, which we refer to as the average information leakage (AIL). The AIL can be defined as
\begin{align}\label{delta_bar}
    \bar{\delta} &= \E\{\delta | h_b\}\nonumber\\
    &\hspace*{-5mm}=\int_{\R^+} \Q\left(\sqrt{\frac{N}{V_e(x)}}\left[\log_2\left(\frac{1+\gamma_b}{1+x}\right) -{R_0}\right]\right)f_{\gamma_e}(x) dx,
\end{align}
where $V_e(x)=\log^2_2\e\frac{x(x+2)}{(x+1)^2}$ and ${R_0}\treq\sqrt{\frac{V_b}{N}}\Qi(\varepsilon)+\frac{m}{N}$.


Note that obtaining a closed-form expression for \eqref{delta_bar} is extremely challenging, due to the complicated integration over the composite Gaussian Q-function. In Proposition \ref{prop1}, we derive an approximate expression for  \eqref{delta_bar} by means of Laplace's approximation theorem \cite{laplace}. 
\begin{prop}\label{prop1}
The AIL for the considered wiretap system with Rayleigh fading channels can be approximated as
 \begin{align}\label{deltaApprox}
      \bar{\delta} &\approx  \exp{\left({-\frac{x_0}{\bgamma_e}}\right)},
 \end{align}
where $x_0$ is defined for $x_0\geq 0 $ as 
\begin{align}
    x_0 = \frac{1+\gamma_b}{2^{R_0}}-1.
\end{align}
 \end{prop}
\begin{proof}
Integrating by parts, we rewrite \eqref{delta_bar} equivalently as
\begin{align}\label{delta_bar_CDF}
   \begin{split}
       \bar{\delta} =\int_{\R^+} &\left[1-F_{\gamma_e}(x)\right]\times\\
       &\frac{\partial }{\partial x}\Q\left(\sqrt{\frac{N}{V_e(x)}}\left[\log_2\left(\frac{1+\gamma_b}{1+x}\right)-{R_0}\right]\right) dx,
   \end{split} 
\end{align}
where $F_{\gamma_e}(\cdot)$ indicates the CDF of the r.v. $\gamma_e$ given by \eqref{gamma}. Because the first derivative of the Q-function is   $\frac{\partial \Q(x)}{\partial x} = -\frac{1}{\sqrt{2\pi}}\exp\left(-\frac{x^2}{2}\right)$, applying the chain rule we can express \eqref{delta_bar_CDF} as
\begin{align}\label{laplace}
      \bar{\delta} =  \int_{\R^+} \Psi(x) \e^{-N\Xi(x)} dx,
\end{align}
where the functions $\Psi(x)$ and $\Xi(x)$ are given respectively by
\begin{align}
   \hspace{-3mm} \Psi(x) \hspace{-0.5mm}=\hspace{-0.5mm} \sqrt{\frac{N}{\pi x(x+2)}}\hspace{-0.5mm}\left(\hspace{-0.5mm}1\hspace{-0.5mm}+\hspace{-0.5mm}\frac{\log_2\left(\frac{1+\gamma_b}{1+x}\right)\hspace{-0.5mm}-\hspace{-0.5mm}{R_0}}{x(x+2)\log_2\e}\hspace{-0.5mm}\right)\hspace{-0.5mm}\left[1\hspace{-0.5mm}-\hspace{-0.5mm}F_{\gamma_e}(x)\right],
    \end{align}
    and
    \begin{align}
    \Xi(x) = \frac{\left(\log_2\left(\frac{1+\gamma_b}{1+x}\right)-{R_0}\right)^2}{2V_e(x)}.
\end{align}
Note that $\Xi(x)$ is a twice-differentiable function whose global minimum occurs at $x_0$, and the function $\Psi(x)$ is smooth. Thus, with all the required conditions satisfied for an application of the Laplace (a.k.a. saddle-point) approximation, we can approximate \eqref{laplace} as
\begin{align} \label{delt1}
    \bar{\delta} &\approx \e^{-N\Xi(x_0)}\Psi(x_0) \sqrt{\frac{2\pi}{N\Xi''(x_0)}}.
\end{align}
Finally, computing the values of the corresponding functions in \eqref{delt1} at $x_0$, i.e., $\Xi(x_0)=0$, $\Psi(x_0)=\sqrt{\frac{N}{\pi x_0(x_0+2)}}[1-F_{\gamma_e}(x_0)]$, and $\Xi''(x_0)=\frac{2}{x_0(x_0+2)}$, as well as considering the CDF of the r.v. $\gamma_e$ given by \eqref{gamma}, some algebraic simplifications lead to the approximate expression for $\bdelta$ given by \eqref{deltaApprox}. This completes the proof.
\end{proof}

\begin{remark}
    Based on Proposition \ref{prop1}, we can conclude that increasing the blocklength $N$ reduces the AIL performance $\bdelta$, benefiting communication secrecy at the cost of reducing the secrecy rate $R_s$. The AIL is always non-zero regardless of how much communication power is allocated.
\end{remark}

\subsection{Secrecy Outage vs. Average Information Leakage}
Recall that a secrecy outage event in the conventional IBL regime, where Eve can observe the infinite blocklength, can be characterized as an event where the capacity of the wiretap link $C_e=\log_2(1+\gamma_e)$ becomes greater than the redundancy rate $R_e$, as in \cite{zhou2011rethinking}. Accordingly, the secrecy outage probability (SOP) is formulated as
\begin{align}\label{pso}
P_{so} &= \Pr\{C_e > R_e\} = 1- F_{\gamma_e}\left(2^{R_e}-1\right),\\
& = \exp\left(\frac{1 - 2^{R_e}}{\bgamma_e}\right).
\end{align}
Comparing \eqref{deltaApprox} and \eqref{pso}, we can see that these two expressions are identical for a particular choice of $R_e$. This provides a nice link between the well-known SOP in the IBL regime and the AIL in the FBL regime from a statistical viewpoint, leading to the important corollary stated below.
\begin{cor}
The Laplace/saddle-point approximation of the AIL in the FBL regime given in {\normalfont{Proposition \ref{prop1}}} is equivalent to the SOP in the IBL regime given by \eqref{pso} when the redundancy rate is chosen as
\[R_e = \log_2(1+\gamma_b) -\sqrt{\frac{V_b}{N}}\Qi(\varepsilon)-\frac{m}{N}.\]
\end{cor}

\section{Blocklength  Design for Secure Transmission}\label{sec:problem}
In this section, we turn our focus to the transmission design problem in which we optimize the coding block length $N$ for the considered secure FBL communication system. Our problem is formulated as a multi-objective optimization problem (MOOP), aiming to minimize the AIL $\bdelta$ while maximizing the effective secrecy throughput  (EST) for the considered system, defined as $\mathcal{T}=\frac{(1-\varepsilon)m}{N}$,  given by
\begin{align}\label{opt_prob_original}
\Prob{}:&~ \{\stackrel{}{\underset{N}{\mathrm{max}}}\mathcal{T},~~\stackrel{}{\underset{N}{\mathrm{min}}}\bdelta\}
\nonumber\\
&~~\text{s.t.}~~ 1 \leq N \leq N^{max},~N\in\mathbb{Z}^{+}
\end{align}
where $N^{max}$ is the maximum blocklength determined by the maximum allowed delay of the FBL transmission. Problem \Prob{} leads to an adaptive blocklength design since Alice possesses knowledge of Bob's instantaneous CSI and hence $\gamma_b$; thus it is assumed that she has the ability to adapt the design parameter $N$ in accordance with this information. Note that we can convert \Prob{} into a single-objective optimization problem  (SOOP) by introducing a weighting factor $\lambda$ ($0\leq\lambda \leq 1$), as follows:  
\begin{align}\label{opt_prob_caseII}
\Prob{1}:&~ \stackrel{}{\underset{N}{\mathrm{min}}}\lambda \bdelta -\frac{ (1-\lambda)\mathcal{T}}{m(1-\varepsilon)}
\nonumber\\
&~~\text{s.t.}~~ 1 \leq N \leq N^{max},~N\in\mathbb{Z}^{+}.
\end{align}
The parameter $\lambda$ determines the relative importance of the AIL and the EST, after proper scaling. Here, we scale the EST by the factor $(1-\varepsilon)m$ to make its range comparable to that of the AIL. Note that if we choose $\lambda=0$, then \Prob{1} reduces to an EST maximization regardless of the amount of information leakage. On the other hand, the extreme case $\lambda=1$ minimizes the AIL without regard for the EST. It can be readily seen that for the former, the optimal blocklength is $N^*=1$, while for the latter $N^*=N^{max}$, due to the monotonicity of the respective objective functions. Nevertheless, for the general case $0<\lambda<1$, the objective function is nonconvex, and its optimal solution can be determined via a one-dimensional (1D) search. Note that this scaling approach results in a set of solutions that form a Pareto-optimal boundary, allowing system designers to trade-off the AIL and EST depending on the given circumstances, as shall be shown below.

\subsection{Alternative Low-complexity Approach}
From a practical perspective, the AIL should generally not exceed a certain threshold $\phi$ depending on the secrecy level required for the application of interest. As such, we can reformulate \Prob{} as a SOOP, by bringing the second objective function into the constraint set, as
\begin{subequations}
\begin{align}\label{opt_prob_caseI}
\Prob{2}:&~ \stackrel{}{\underset{N}{\mathrm{max}}}\mathcal{T}
\nonumber\\
&~~\text{s.t.}~~ 1 \leq N \leq N^{max},~N\in\mathbb{Z}^{+}\\
&~~\qquad \bdelta \leq \phi.\label{sec_cst}
\end{align}
\end{subequations}
Note that the value of the objective function in \Prob{2} increases as $N$ decreases, while the AIL $\bdelta$ is monotonically non-decreasing, leading to a tightening of the constraint \eqref{sec_cst}. Thus, the optimal solution to \Prob{2} is the minimum value of $N$ for which the constraint \eqref{sec_cst} is satisfied. As a result, assuming the feasibility of \Prob{2} and using \eqref{deltaApprox}, we can write the optimality condition as
\begin{align}\label{opt_eq}
    \log_2\left(\frac{1+\gamma_b}{1-\bgamma_e \ln \phi}\right) = \frac{\sqrt{V_b}\Qi(\varepsilon)}{\sqrt{N}}+\frac{m}{N},
\end{align}
Applying the change of variable $N=\eta^2$, \eqref{opt_eq} can be transformed into the quadratic equation 
\[a\eta^2-b\eta-m=0,\]
where the coefficients are given respectively by 
\[a=\log_2\left(\frac{1+\gamma_b}{1-\bgamma_e \ln \phi}\right)\textrm{and}~b = \sqrt{V_b}\Qi(\varepsilon).\]
Retaining the positive solution to the quadratic equation leads to the optimal integer-valued blocklength 
\begin{align}
    N^* = \min\left\{\Bigg\lceil \sqrt{\frac{b+\sqrt{b^2+4am}}{2a}}\Bigg\rceil, N^{max}\right\},
\end{align}
where $\lceil x\rceil$ indicates the smallest integer larger than or equal to $x$ and is used to ensure the constraint \eqref{sec_cst} is not violated. Consequently, the optimal EST is given by $\mathcal{T}^*=\frac{(1-\varepsilon)m}{ N^*}$.

\section{Numerical Results and Discussion}\label{sec:simulations}
In this section, we provide simulation results to validate the accuracy of the approximate expressions obtained above for the AIL and investigate the impact of key system parameters such as blocklength, number of confidential bits, transmit power, and decoding error probability on the AIL and EST performance. Unless otherwise stated, the simulation parameters are set as follows: Number of transmit information bits $m=200$, transmit SNR $\rho=0$ dB, decoding error probability $\varepsilon = 10^{-3}$, maximum tolerable AIL $\phi=10^{-2}$,  coding blocklength $N=400$, maximum blocklength $N^{max}=1000$, $\mu_b = 1$, $ \gamma_b=1$, and $\mu_e= 0.1$.

\begin{figure}[t]
  \centering
  \includegraphics[width=\linewidth]{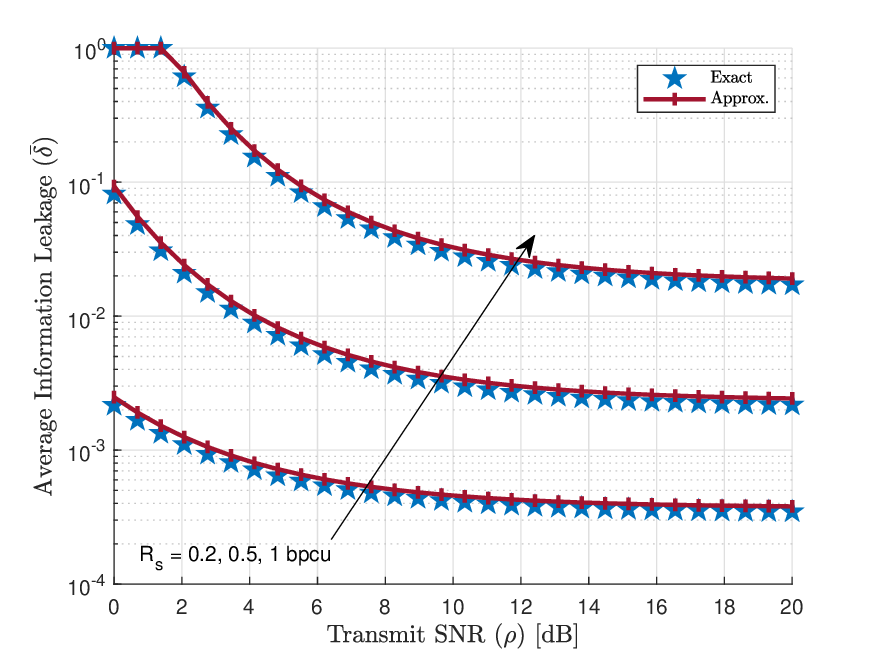}
  \caption{AIL vs. transmit SNR for different $R_s$.}
\label{sim1}
\end{figure}
Fig. \ref{sim1} shows the AIL under Rayleigh fading for a fixed main channel realization and different secrecy coding rates ranging from $R_s=0.2$ bpcu (relatively low) to $R_s=1$ bpcu (relatively high). The markers labeled with \textit{Exact} represent the exact AIL evaluated by \eqref{delta_bar}, while the curves labeled as \textit{Approx.} indicate the approximate analytical evaluation of the AIL in the FBL regime according to Proposition \ref{prop1}. We can see that the theoretical exact and analytical approximate AILs match well for a wide range of transmit SNRs, validating our analysis. In addition, we see that for a given transmit SNR, a higher secrecy rate $R_s$ corresponds to a larger AIL.  Furthermore, the high-SNR performance reveals an AIL floor regardless of $R_s$, 
where the AIL approaches a nonzero value for large $\rho$. This indicates that allocating more resources to increase $\rho$ has a negligible impact on the AIL, signifying the need for proper resource management. Overall, this figure confirms the accuracy of our derived approximate expressions for the AIL and once again reveals the inherent close relationship between the  AIL in the FBL regime and the SOP for IBL.

\begin{figure}[t]
\centering
\includegraphics[width=\columnwidth]{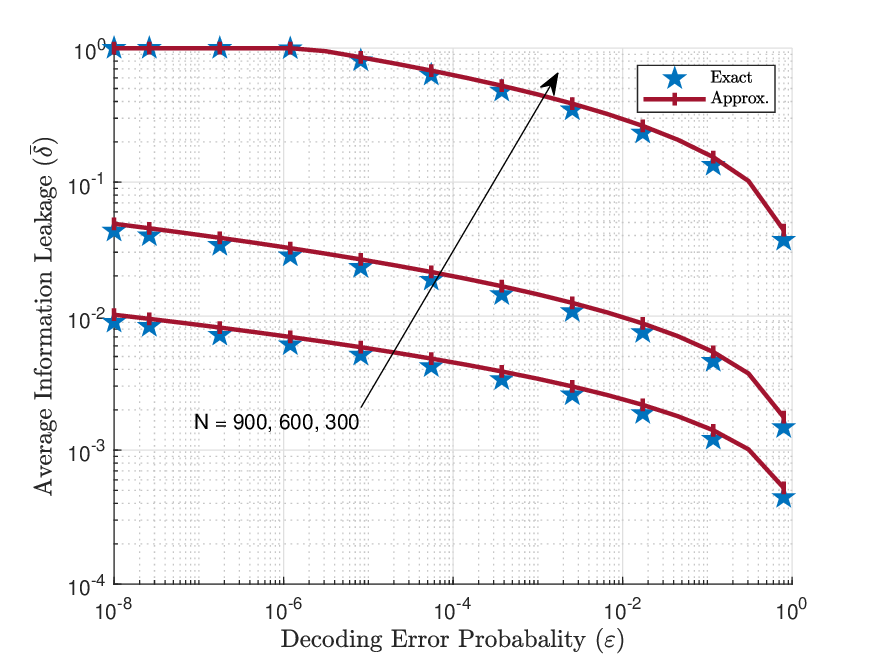}
\caption{AIL vs. decoding error probability for different $N$.}
\label{sim3}
\end{figure}

Fig. \ref{sim3}  provides notable insight into the security-reliability trade-off for the considered secure FBL transmission system. In particular, we plot the AIL against the decoding error probability $\varepsilon$ for different blocklengths. As can be clearly seen from the figure, the exact and approximate AIL are in good agreement for different values of $\varepsilon$, confirming the correctness of our analysis. Furthermore, when $\varepsilon$ increases, the AIL becomes smaller, and this decreasing trend becomes more significant as the coding blocklength is reduced. We can also conclude from the figure that by compromising a small level of secrecy in terms of AIL, a large improvement in reliability is achieved. For example, when $N=900$, if the secrecy level is reduced from $\bdelta = 2 \times 10^{-3}$ to $\bdelta =10^{-2}$, which corresponds to roughly $7$ dB security loss, the communication reliability is improved from $\varepsilon= 10^{-2}$ to $\varepsilon=10^{-8}$, i.e., approximately a $60$ dB gain in reliability. This is particularly beneficial for systems that require ultra-high reliability and can tolerate a somewhat reduced level of secrecy.

\begin{figure}[t]
\centering
\includegraphics[width=\columnwidth]{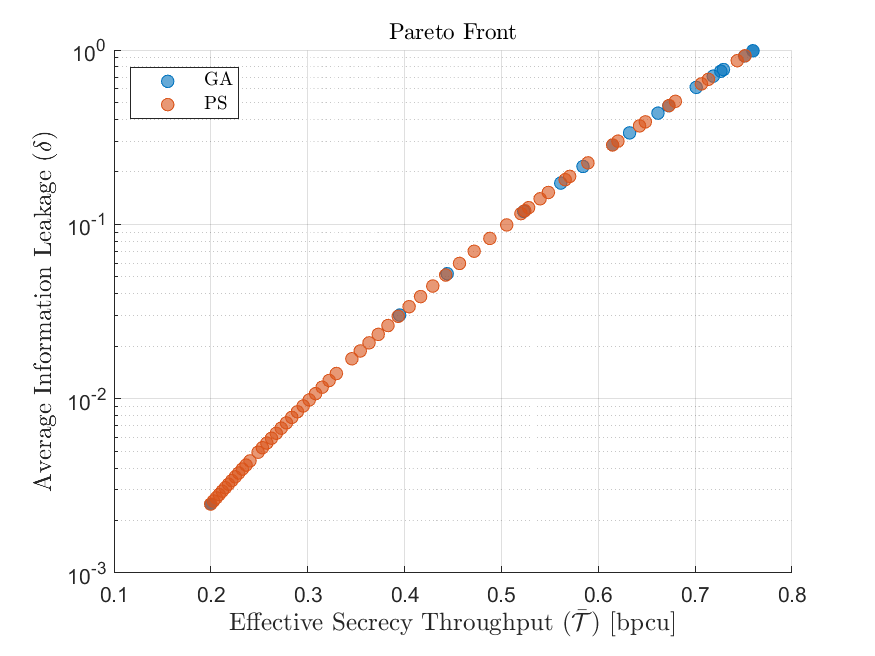}
\caption{Pareto boundary for the MOOP given by \Prob{}.}
\label{sim7}
\end{figure}
Fig. \ref{sim7} shows the Pareto boundary for the MOOP problem, illustrating how improvement in one objective results in degradation for the other. 
Variations of two algorithms are employed to solve \Prob{} via the Matlab Optimization Toolbox \cite{OptimizationToolbox}, namely the genetic algorithm (GA) and the direct/Pareto search (PS) algorithm. Both algorithms perform well in finding the optimal points. PS finds a more concentrated set of points on the boundary, while the GA solutions occur more at the extreme ends of the performance metrics. Overall, this figure provides useful insight into the set of optimal AIL and EST values and their trade-offs. 

\begin{figure}[t]
\centering
\includegraphics[width=\columnwidth]{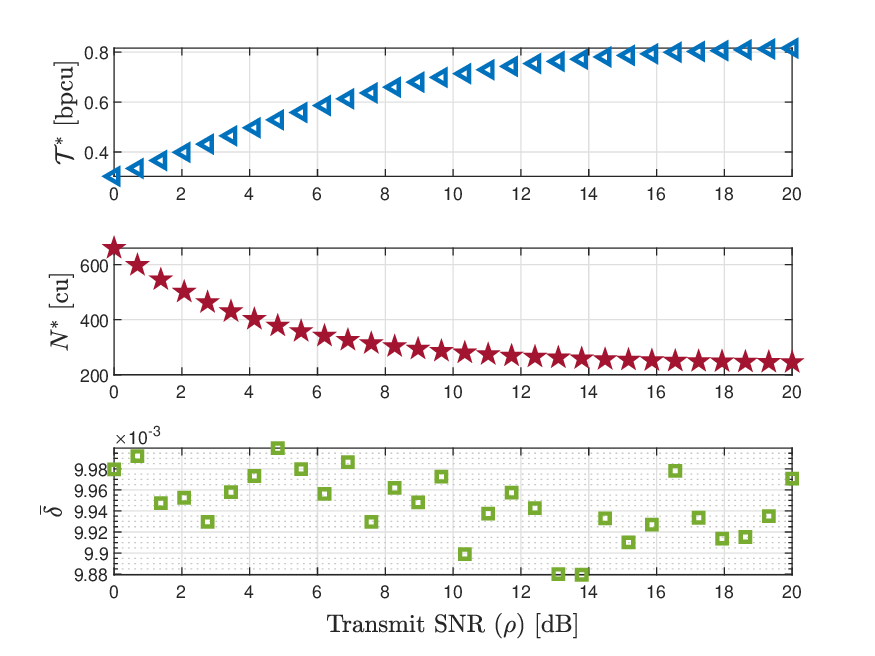}
\caption{Designed EST and blocklength vs. transmit SNR using \Prob{1}.}
\label{sim5}
\end{figure}

Fig. \ref{sim5} depicts the optimal EST and designed blocklength $N^*$, as well as the obtained AIL for different transmit SNRs. It is clear from the figure that the optimal blocklenth is monotonically increasing with transmit SNR $\rho$, yet experiences a ceiling phenomenon at large SNRs due to the maximum transmission delay requirement corresponding to $N^{max}$. Furthermore,  we observe from Fig. \ref{sim5} that the larger the transmit SNR, the smaller the number of channel uses that is optimally required for EST improvement while satisfying the AIL security requirement.     
\begin{figure}[t]
\centering
\includegraphics[width=\columnwidth]{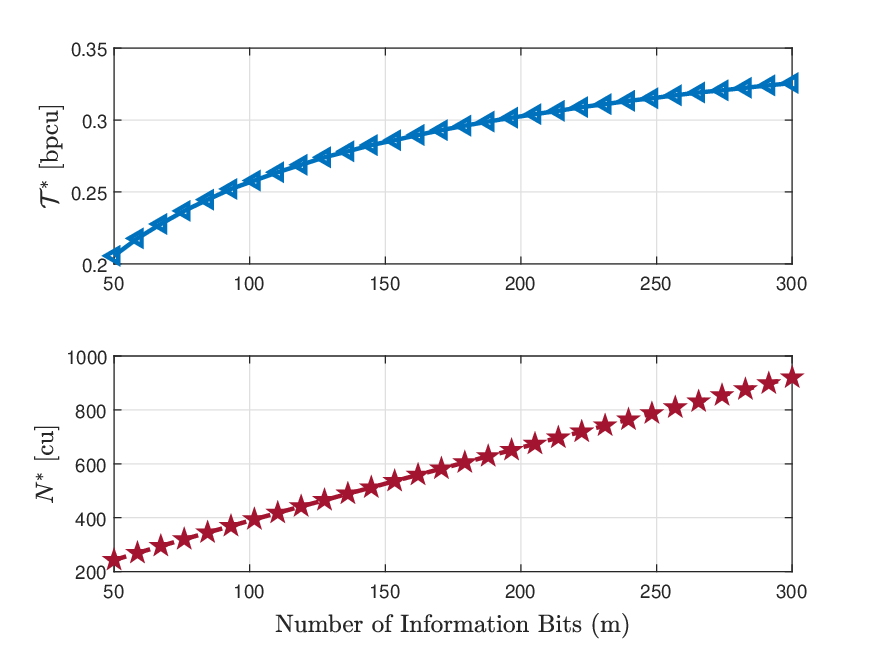}
\caption{Designed EST and blocklength vs. the number of information bits via \Prob{1}.}
\label{sim6}
\end{figure}

Fig. \ref{sim6} exhibits the impact of the number of transmit information bits $m$ on the designed parameter $N^*$  as well as the optimal EST $\bar{\mathcal{T}}^*$ according to \Prob{1}. We observe from the figure that as $m$ increases, $N^*$ grows linearly to improve the EST for maintaining the demanded secrecy level. Furthermore, the overall EST performance $\bar{\mathcal{T}}^*$ is non-decreasing with $m$; however, due to the maximum blocklength threshold $N^{max}$, a ceiling phenomenon on the optimal EST appears when $m$ becomes sufficiently large.



\section{Conclusion}\label{sec:conclusion}
This work has explored the performance analysis and optimization of secure FBL transmissions in terms of the AIL metric for a wiretap communication system without instantaneous knowledge of Eve's CSI.  We have obtained expressions for the exact and approximate AIL of the system assuming Rayleigh fading channels. We have unveiled the relationship between the AIL in the FBL regime and the SOP in the IBL regime and found a simple relationship between these two statistical metrics. We have further investigated the impact of key system parameters on the AIL performance. Our findings have revealed that increasing blocklength reduces the system's AIL, and sacrificing a small level of secrecy could lead to a substantial improvement in reliability. Finally, we have formulated and solved a blocklength optimization problem to improve the secrecy performance of the system using different optimization approaches, obtaining worthwhile insights into the practical design of PLS with FBL transmissions.

\section*{Acknowledgment}
This work was supported by the Australian Research Council’s Discovery Projects (project number DP220101318).

\bibliographystyle{IEEEtran}
\bibliography{RefList}

\end{document}